\documentclass[letterpaper,twocolumn,10pt]{article}
\usepackage{usenix}
\usepackage{tabularx,tabulary}
\usepackage{booktabs}
\usepackage{algorithmicx,algpseudocode, algorithm}
\usepackage{amsthm}
\usepackage{tikz}
\usepackage{amsmath}

\newtheorem{theorem}{Theorem}[section]

\newtheorem{lemma}[theorem]{Lemma}

\usepackage{algpseudocode} 
\usepackage{algorithm}
\usepackage{tabularx} 
\usepackage{booktabs}
\usepackage{enumitem}

\usepackage{filecontents}
\usepackage{subcaption} 
\usepackage{enumitem}

\usepackage[subtle]{savetrees}

\begin{document}

\date{}

\title{\Large \bf Neighborhood Blending: A Lightweight Inference-Time Defense Against Membership Inference Attacks
}

\author{
{\rm Osama Zafar}\\
Case Western Reserve University
\and
{\rm Shaojie Zhan}\\
Texas Tech University
\and
{\rm Tianxi Ji}\\
Texas Tech University
\and
{\rm Erman Ayday}\\
Case Western Reserve University
} 

\maketitle

\begin{abstract}
In recent years, the widespread adoption of Machine Learning as a Service (MLaaS), particularly in sensitive environments, has raised considerable privacy concerns. Of particular importance are membership inference attacks (MIAs), which exploit behavioral discrepancies between training and non-training data to determine whether a specific record was included in the model's training set, thereby presenting significant privacy risks. Although existing defenses, such as adversarial regularization, DP-SGD, and MemGuard, assist in mitigating these threats, they often entail trade-offs such as compromising utility, increased computational requirements, or inconsistent protection against diverse attack vectors.

In this paper, we introduce a novel inference-time defense mechanism called Neighborhood Blending, which mitigates MIAs without retraining the model or incurring significant computational overhead. Our approach operates post-training by smoothing the model's confidence outputs based on the neighborhood of a queried sample. By averaging predictions from similar training samples selected using differentially private sampling, our method establishes a consistent confidence pattern, rendering members and non-members indistinguishable to an adversary while maintaining high utility. Significantly, Neighborhood Blending maintains label integrity (zero label loss) and ensures high utility through an adaptive, "pay-as-you-go" distortion strategy. It is a model-agnostic approach that offers a practical, lightweight solution that enhances privacy without sacrificing model utility. Through extensive experiments across diverse datasets and models, we demonstrate that our defense significantly reduces MIA success rates while preserving model performance, outperforming existing post-hoc defenses like MemGuard and training-time techniques like DP-SGD in terms of utility retention.

\end{abstract}

\section{Introduction}
\label{sec:intro}

In recent years, machine learning (ML) has become a cornerstone of modern computing, facilitating substantial advancements across a wide range of applications. From image and speech recognition to natural language processing, recommendation systems, healthcare diagnostics, and autonomous systems, ML models are increasingly integrated into both consumer-oriented and mission-critical services. The availability and accessibility of large-scale datasets, high-performance computing, and the development of sophisticated learning algorithms have fueled this widespread adoption. As a result, ML models are often trained on vast amounts of sensitive data and deployed as services accessible through public or semi-public interfaces.

Although the exceptional performance of these models has driven significant technological advances, it also introduces security and privacy concerns. Unlike traditional software, ML models embed training data into their parameters and outputs, which adversaries can exploit to extract sensitive information. Multiple studies have shown that ML models tend to memorize sensitive information from training data, posing serious privacy risks \cite{Carlini2019,Fredrikson2015,Ganju2018,Hitaj2017,Shokri2017,Song2017}. Adversaries may exploit this memorization by leveraging the differences in the model’s behavior on training and non-training data to determine whether a specific data record was included in the training set of a target model \cite{Shokri2017,yeom2018}. Such privacy attacks, known as membership inference attacks (MIAs), pose a significant privacy threat, as knowledge of whether an individual’s data was used for training can itself disclose sensitive information about that individual \cite{pyrgelis2017,Backes2016}. Since their formal introduction, MIAs have become a canonical privacy threat, particularly for models trained on sensitive or proprietary data. Studies have shown that ML classifiers are vulnerable to MIAs \cite{Nasr_2019,salem2018,Shokri2017,Song2019}. MIAs are especially relevant in deployed ML services such as ML-as-a-Service (MLaaS) \cite{Ribeiro2015MLaaSML}, where adversaries can interact with models via prediction APIs but lack direct access to model parameters or training data. 

To mitigate the privacy risks posed by MIAs, prior studies have introduced various defense strategies aimed at minimizing information leakage during model training or inference. Notably, Jia et al.~\cite{Jia2019memgaurd} developed a post hoc solution called MemGuard, which adds adversarial noise to the confidence vector that is provided as a part of the model's output to obscure the membership signal. This method operates during inference, requires no data retention, and is model-agnostic. Likewise, Nasr et al.~\cite{Nasr2018} proposed a training-time defense that balances predictive accuracy with low inference risk by incorporating adversarial regularization. Several other defense mechanisms have been proposed to mitigate privacy risks associated with membership inference; however, empirical evidence shows that these approaches fail to achieve acceptable privacy-utility tradeoffs for complex models~\cite{Jayaraman2019}. In particular, a comprehensive systematic evaluation study by Song et al.~\cite{Song2020SystematicEO} (see Table~\ref{tab:dataset_defense_comparison}) demonstrated that numerous existing defenses provide robustness within their presumed attack scenarios but tend to lack generalizability when evaluated against a diverse range of other attack methodologies. In practice, acceptable trade-offs exist only in narrow, simple settings, and not for complex models. The existing methods tend to fall into one of two regimes: models are either accurate but not private, or private but not accurate. To achieve accuracy anywhere close to non-private models, the privacy budget must be increased to a degree where the mathematical privacy guarantee effectively becomes non-existent.


\begin{table}[htbp]
\centering
\caption{Membership Inference Attack (MIA) accuracy across datasets under different defense mechanisms. Lower accuracy indicates a stronger defense.}
\label{tab:dataset_defense_comparison}
\small
\begin{tabularx}{\columnwidth}{l X c c c}
\toprule
\textbf{Dataset} & \textbf{Defense} & \textbf{Song et al. \cite{Song2020SystematicEO}}\\
& \textbf{Mechanism} & \textbf{Benchmark}\\
&  & \textbf{Attack Acc.}\\
\midrule
\textbf{Purchase-100} & Adv. Reg. \cite{Nasr2018} & 59.5\%\\
        & Neighborhood-Defense & \textbf{51.1\%} \\
\addlinespace
\textbf{Location-30}  & MemGuard \cite{Jia2019memgaurd}  & 69.1\%\\
        & Neighborhood-Defense  & \textbf{49.7\%} \\
\addlinespace
\textbf{Texas-100}    & Adv. Reg. \cite{Nasr2018}  & 58.6\%\\
        & MemGuard \cite{Jia2019memgaurd}  & 74.2\%\\
        & Neighborhood-Defense  & \textbf{49.4\%} \\
\bottomrule
\end{tabularx}
\end{table}

In this paper, we introduce \textit{Neighborhood Blending}, a novel inference-time defense mechanism designed to protect machine learning models against black-box membership inference attacks (MIA). Unlike traditional defenses that modify the training process, Neighborhood Blending operates post-training by smoothing the model’s confidence outputs based on the neighborhood of a queried sample within the training dataset. By averaging prediction vectors from nearby samples with the same predicted class, our method's smoothed output maintains the original class prediction, guaranteeing zero label loss, while simultaneously masking the subtle behavioral differences that MIAs exploit. This effectively renders the model’s responses statistically indistinguishable and establishes a consistent confidence pattern for both members and non-members, thereby exhibiting uniform behavior from the attacker’s perspective. Across a wide range of datasets and models, Neighborhood Blending consistently reduces membership inference accuracy from as high as $0.95$ to near the random-guessing baseline ($\approx0.50$). When evaluated under the diverse attack benchmarks introduced by Song et al.~\cite{Song2020SystematicEO}, our approach outperforms existing defense mechanisms such as MemGuard and Adversarial Regularization, demonstrating robust protection against both learned and metric-based membership inference attacks. It also preserves model utility by guaranteeing zero label loss and very low confidence vector distortion (see Table~\ref{tab:confidence_differences}). Compared to DP-SGD it provides higher model accuracy and order of magnitude lower confidence vector distortion (Table~\ref{tab:sgd_dpsgd_confidence_differences}). Neighborhood Blending is a practical and lightweight defense that can be applied to any pre-trained classifier without retraining (see details in Section~\ref{sec:evaluation}).


\section{Related Work}
\label{sec:relatedwork}

\subsection{Membership Inference Attacks}

Membership inference attacks (MIAs) aim to determine whether a specific data record was included in a model's training dataset. They exploit systematic differences in a model’s behavior on training versus non-training samples, enabling adversaries to infer membership from model outputs alone. Prior work has demonstrated the effectiveness of MIAs across a wide range of data modalities and learning tasks, as well as the challenges involved in designing effective defenses.

Formally, consider a classification model, $T$ trained on a dataset $D_t$, which consists of sample data and corresponding labels. The primary objective of this model is to accurately learn how to classify input data into predefined categories. When presented with a query sample, $q$, the model $T$ generates a confidence vector, $V_q$. This vector represents the likelihood that the sample belongs to each class, with the probabilities summing to 1. The predicted class label is determined by the class corresponding to the maximum probability in $V_q$, represented by $\arg\max(V_q)$. An attacker, $A$, queries the model, $T$, with random samples, and receives output in the form of confidence vectors $V_q$. $A$ examines how the model reacts to different samples to decide if they are members of the training dataset or not.



Shokri et al. \cite{Shokri2017} initially formalized membership inference attacks by showing how an adversary can determine if a data sample was part of a target black-box classification model's training set. In their approach, the attacker trains one or more binary attack models that receive confidence score vectors from the target model and decide membership. To develop these attack models, the adversary employs a shadow-training strategy by training multiple shadow models on auxiliary datasets derived from the same distribution as the training data of the target model. Since the attacker has control over the shadow models, they are able to distinguish between member and non-member samples, thereby enabling supervised training of the attack models using their prediction outputs.

Salem et al. \cite{salem2018} expanded this research, showing many of Shokri et al.'s assumptions can be relaxed without greatly impacting attack success. They demonstrated that often only one shadow model is needed, which does not have to match the target model's architecture. They also proved that auxiliary data from different distributions can still enable successful attacks. These findings lower the cost of membership inference attacks and suggest black-box ML models are more vulnerable than previously thought.

Subsequent work demonstrated that membership inference remains effective even when target models are well-generalized. Long et al. \cite{long2020pragmatic} showed that although average-case inference accuracy may be low, certain training records, such as outliers, are significantly more vulnerable, enabling high-precision inference when attacks selectively target them. Carlini et al. \cite{carlini2022Lira} proposed the Likelihood Ratio Attack (LiRA), which frames membership inference as a statistical hypothesis test by comparing the target model’s confidence against distributions learned from shadow models trained with and without the target sample, enabling accurate inference at low false positive rates. Beyond black-box settings, Nasr et al. \cite{Nasr_2019} introduced membership inference attacks in the white-box setting by leveraging gradients of the target model with respect to individual data samples as inference features. They further extended these techniques to federated learning, showing that intermediate training updates can leak membership information. Similarly, Melis et al. \cite{melis2018exploitingunintendedfeatureleakage} demonstrated that participants in federated learning can infer the membership of other users’ data by analyzing shared model updates, highlighting that distributed training protocols are also susceptible to membership leakage.
In parallel, it has been shown that a class of metric-based attacks avoids training explicit attack models altogether and instead relies on simple statistics derived from the target model’s outputs. Yeom et al. \cite{yeom2018} proposed inferring membership based on prediction correctness or prediction loss, exploiting the fact that models typically achieve lower loss on training samples. Salem et al. \cite{salem2018} further explored confidence-based and entropy-based metrics, showing that training samples often yield higher confidence and lower entropy predictions than non-members. Song et al. \cite{Song2020SystematicEO} refined entropy-based attacks by incorporating ground-truth labels, mitigating errors caused by confidently incorrect predictions.

Overall, these attacks highlight that simply having black-box access is enough to successfully perform membership inference across different models and deployment settings. Our focus is on black-box models that provide confidence scores, which are the most common in practical machine learning today. These scores are important for transparency because they help users understand how reliable a prediction is; at the same time, they can also be used by attackers as a key signal. The various attack methods, from shadow-model techniques to simple metric-based approaches, show that membership leakage is a fundamental challenge in machine learning systems. They also demonstrate how difficult it can be to protect against these attacks quickly and with ready-made solutions.

\subsection{Mitigation Techniques}

Various defense strategies are being developed to help protect against membership inference attacks. These approaches aim to reduce information leakage during model training or inference and make it harder to tell apart training samples from others, creating a safer environment for data privacy. Early approaches focus on minimizing overfitting through regularization techniques such as L2 regularization \cite{Shokri2017} and dropout \cite{Srivastava2014dropout}, which serve to narrow the confidence disparities between members and non-members. Nasr et al. \cite{Nasr2018} introduced an adversarial training approach that conceptualizes defense as a min–max optimization problem, with the objective of minimizing classification errors while enhancing resistance to membership inference. Furthermore, ensemble defenses such as model stacking \cite{salem2018} contribute to reducing memorization by training multiple classifiers on distinct datasets and aggregating their outputs. Differential privacy \cite{Dwork2006DP} provides the strongest formal protection by injecting noise into the training process or gradients to bound the influence of individual training samples \cite{Abadi_2016, bassily2014differentiallyprivateempiricalrisk,chaudhuri2011differentiallyprivateempiricalrisk,Song2013StochasticGD,wang2018differentiallyprivateempiricalrisk,Yu_2019}, and has been applied to collaborative and federated learning settings \cite{Shokri2015privacyDL}. More targeted defenses have been proposed, such as adversarial regularization \cite{Nasr2018} that integrates a membership inference adversary into the training process, formulating a min–max optimization that jointly minimizes classification loss while misleading an attack classifier. 

Beyond training-time defenses, inference-time mechanisms have also been explored. Jia et al. \cite{Jia2019memgaurd} proposed MemGuard, a post-training defense that obfuscates model outputs by adding carefully crafted perturbations to confidence score vectors, thereby reducing an adversary’s ability to distinguish between members and non-members without retraining the target classifier. While this approach avoids modifying the training process, it is not entirely training-free or parameter-free, as it relies on auxiliary training of a defense model and per-output noise optimization.

\noindent\textbf{Core limitations:} Training-time regularization-based defenses can offer robust protection against membership inference attacks; however, they often incur substantial utility loss and computational overhead because they require retraining deployed models. In contrast, MemGuard avoids retraining and provides robust protection by applying inference-time noise addition to model outputs. Nevertheless, systematic evaluations by Song et al. \cite{Song2020SystematicEO} reveal that such defenses do not offer uniform robustness across different classes of membership inference attacks (see Table~\ref{tab:dataset_defense_comparison}). In particular, while MemGuard effectively reduces the success of shadow-model-based attacks, it remains vulnerable to metric-based attacks that exploit simple statistical characteristics of model outputs, such as confidence and entropy. These shortcomings underscore the necessity for more comprehensive and pragmatic defense mechanisms.

\section{System Settings}

\label{sec:setting}
In this section, we introduce the system model and the threat model.

\subsection{System Model}
\label{sec:system}

In our system model, we consider three parties: the model provider, the attacker, and the defender. These roles may be assumed by distinct entities or consolidated within a single entity. For example, in certain deployments, the model provider may also act as the defender.

\noindent\textbf{Model Provider:} We consider a model provider with a sensitive training dataset, such as healthcare or finance data. The provider trains a classifier model using this proprietary data and then deploys the model for inference. After training, the model is accessed via an inference interface, such as a prediction API, which provides classification outputs, including confidence score vectors and predicted labels for the input data. The provider's primary goal is to maximize the model's utility and ensure its availability.

\noindent\textbf{Attacker:} The attacker aims to infer the model provider's sensitive training dataset. We assume the attacker has only black-box access. The attacker leverages black-box membership inference attacks \cite{long2018understandingmembershipinferenceswellgeneralized, Nasr2018, salem2018, Shokri2017} to infer whether a given data record was included in the target model's training dataset. The attacker can interact with the deployed model by issuing queries and observing its outputs, but does not have direct access to the model’s parameters, architecture, or training data. The attacker’s goal is to achieve membership inference with accuracy significantly better than random guessing.

\noindent\textbf{Defender:} The defender seeks to protect the privacy of the training data from membership inference attacks. The defender may be the model provider itself or an independent entity responsible for enforcing privacy protections. They can use different defense mechanisms \cite{Shokri2017,Srivastava2014dropout,Nasr2018,salem2018,Song2013StochasticGD,wang2018differentiallyprivateempiricalrisk,Yu_2019} that can be applied either during model training, for example, through regularization or privacy-aware optimization, or at inference time, such as by perturbing model outputs or obfuscating confidence scores. Overall, the goal of the defender is twofold: (i) effectively prevent the attacker from inferring whether a data record was included in the training dataset, and (ii) minimize the loss in model utility introduced by the defense mechanism.

\subsection{Threat Model}
\label{sec:threat}

We consider a practical threat model that closely resembles real-world deployments, such as Machine Learning as a Service (MLaaS). In this context, the adversary operates with an unrestricted query budget and observes model responses for various inputs. However, they are not authorized to modify the model, influence its training process, or access internal details such as gradients or parameters. The adversary may also adapt their queries based on the information acquired from the outputs, thereby devising a more refined and strategic approach.

We assume that the attacker is aware of the general learning task, the output format of the model, and the existence of membership inference defenses. However, the attacker lacks knowledge of the exact training dataset. The defender is presumed to be aware of the potential for membership inference attacks and may implement countermeasures, yet lacks information about the specific queries used by the attacker. A robust adversary is assumed to be capable of employing state-of-the-art membership inference techniques, including adaptive and data-driven attacks. An effective defense is defined as one that substantially diminishes the attacker’s membership inference accuracy without causing prohibitive impairments to the model's utility.

\section{Proposed Solution}
\label{sec:workflow}

MIAs exploit the discrepancies in the model’s behavior between members and non-members of the training samples. Ideally, with a fully generalized model, there should be no difference in the confidence outputs regardless of whether the samples are member or non-member. However, this is not always true; models can develop biases from the data or overfit by memorizing training samples rather than learning patterns. In such cases, models behave differently toward known and unknown samples. These differences, such as more peaked probability distributions, lower entropy, or higher confidence margins, allow adversaries to determine if an input was part of the training data. As shown in previous research \cite{yeom2018,salem2018,Song2020SystematicEO,choquettechoo2021labelonlymembershipinferenceattacks,rahimian2020samplingattacksamplificationmembership,Hui_2021}, these subtle yet reliable cues are sufficient for both threshold-based and classifier-based membership inference attacks. The key theoretical insight is that membership distinguishability depends not on the absolute prediction value, but on its deviation from the patterns of nearby training samples. When a model is tested with an input close to a training example, its output usually reflects the confidence profile of its neighboring samples. Conversely, when tested with a non-member input outside this dense region of known samples, the model’s confidence landscape becomes more irregular, revealing membership weaknesses. 

Using this insight, we propose Neighborhood Blending, a post-treatment defense mechanism to mitigate such attacks. Our method establishes a defense by deliberately masking such discrepancies. For any queried sample, we aim to position the query within a neighborhood of actual training samples, at least from the attacker’s perspective. This ensures the prediction accurately reflects the stable and smooth confidence patterns characteristic of that local region. Additionally, we hide subtle behavioral cues and the `nuanced differences' that MIAs exploit, which might otherwise distinguish members from non-members.

By hiding the sample among its neighbors and releasing only smoothed, neighborhood-based behavior, our methodology effectively forces the attacker's observation to collapse into a general, class-consistent pattern that looks statistically identical, regardless of whether the specific data they queried was a genuine part of the original training set (member) or not (non-member). By doing this, we ensure the attacker gains no valuable information that would allow them to distinguish between the two.

\subsection{Methodology}

Let $\mathcal{D}_t=\{(\mathbf{x}_i,y_i)\}_{i=1}^n$ denote the private training dataset and $\mathcal{T}$ be the trained target classifier, mapping inputs $\mathbf{x}\in\mathcal{X}$ to confidence vectors $\mathbf{V}_\mathbf{x} = \mathcal{T}(\mathbf{x})$. Given a query sample $\mathbf{q}$, the confidence vector $\mathbf{V}_{\mathbf{q}}$ and the predicted class $\hat{y}_{\mathbf{q}}$ are determined as follows, where $C$ denotes the total number of classes:
\begin{align*}
\mathbf{V}_{\mathbf{q}} &= \mathcal{T}(\mathbf{q}) \in [0, 1]^C \\
\hat{y}_{\mathbf{q}} &= \underset{c \in \{1, \dots, C\}}{\operatorname{argmax}} (\mathbf{V}_{\mathbf{q}})_c
\end{align*}


The core objective of the defense is to construct a representative neighborhood $\mathcal{N}(\mathbf{q})$ by selecting the $m$ training samples that are semantically closest to the query $\mathbf{q}$. A standard baseline for selecting $m$ training samples is the naive top-$k$ selection. This method computes utility scores such as negative $L_p$ distances for all available candidates and deterministically retrieves the top-$m$ instances. While this greedy strategy maximizes the semantic relevance of the selected neighbors, its deterministic nature lacks privacy guarantees because the inclusion of specific training points is entirely predictable. Conversely, uniform random sampling satisfies privacy but severely degrades utility by ignoring feature similarity. 

To bridge this gap, we model the selection of $m$ neighbors as a single execution of the exponential mechanism. This formulation balances utility and privacy by assigning higher sampling probabilities to candidates closer to the query while introducing necessary stochasticity. Crucially, to preserve the semantic consistency of the prediction, we restrict the candidate set $\mathcal{S}$ to training samples that share the same predicted label as $\mathbf{q}$. 

To prioritize samples that are semantically similar to the query, we assign a utility score to each candidate $\mathbf{x}_i \in \mathcal{S}$. This score is defined as the negative $L_p$ distance to the query:
\begin{equation}
    u(\mathbf{x}_i, \mathbf{q}) = -\|\mathbf{x}_i - \mathbf{q}\|_p.
    \label{eq:utility}
\end{equation}
These scores quantify the relevance of each candidate and drive the exponential mechanism, ensuring that neighbors closer to the query have exponentially higher probabilities of selection, thereby balancing utility with privacy.
The detailed procedure is formalized in Algorithm \ref{alg:dp_defense_set}.

\begin{algorithm}[!htbp]
\caption{Differentially Private Neighborhood Blending}
\label{alg:dp_defense_set}
\begin{algorithmic}[1]
\State \textbf{Input:} Private training dataset $\mathcal{D}_t$, trained classifier $\mathcal{T}$, query $\mathbf{q}$, subset size $m$, privacy budget $\epsilon$.

\Statex \textbf{Initialization:}
\State Compute confidence vector: $\mathbf{V}_{\mathbf{q}} \leftarrow \mathcal{T}(\mathbf{q})$
\State Determine predicted label: $\hat{y}_{\mathbf{q}} \leftarrow \underset{c}{\operatorname{argmax}} (\mathbf{V}_{\mathbf{q}})_{c}$

\Statex \textbf{Label-Conditional Candidate Selection:}
\State Construct candidate set ensuring Zero Label Loss:
\[
\mathcal{S} \leftarrow \{ \mathbf{x}_i \in \mathcal{D}_t \mid \underset{c}{\operatorname{argmax}} (\mathcal{T}(\mathbf{x}_i))_c = \hat{y}_{\mathbf{q}} \}
\]

\Statex \textbf{Utility Scoring:}
\State Compute utility scores $u(\mathbf{x}_i, \mathbf{q})$ for all $\mathbf{x}_i \in \mathcal{S}$ according to Eq. (\ref{eq:utility}).

\Statex \textbf{Set-Based Exponential Sampling:}
\State Compute unnormalized logits $\phi_i$ for all $\mathbf{x}_i \in \mathcal{S}$ according to Eq. (\ref{eq:logits}).
\State Sample Gumbel noise $G_i \sim \text{Gumbel}(0, 1)$ for all $\mathbf{x}_i \in \mathcal{S}$.
\State Compute perturbed scores: $s_i \leftarrow \phi_i + G_i$.
\State Select neighbors $\mathcal{N}(\mathbf{q})$ corresponding to the top-$m$ scores:
\[
\mathcal{N}(\mathbf{q}) \leftarrow \{ \mathbf{x}_i \in \mathcal{S} \mid s_i \text{ is among the top-} m \text{ values} \}
\]

\Statex \textbf{Confidence Smoothing:}
\State Compute the smoothed vector:
\[
\mathbf{V}_{\mathbf{q}}^{\text{avg}} \leftarrow \frac{1}{m} \sum_{\mathbf{x}_j \in \mathcal{N}(\mathbf{q})} \mathcal{T}(\mathbf{x}_j)
\]

\State \textbf{Output:} Smoothed confidence vector $\mathbf{V}_{\mathbf{q}}^{\text{avg}}$.
\end{algorithmic}
\end{algorithm}
\subsection{Efficient Implementation via Gumbel-Top-k Trick}
\label{sec:implementation}

Directly sampling from the exponential mechanism distribution $P(\mathcal{N}) \propto \exp(\frac{\epsilon U(\mathcal{N})}{2\Delta u})$ requires computing the partition function over all possible subsets, which is combinatorially intractable. To circumvent this, we leverage the gumbel-top-$k$ trick. This stochastic optimization technique perturbs the log-probabilities (logits) with i.i.d. gumbel noise. Specifically, we compute the unnormalized logit for each candidate $\mathbf{x}_i$:
\begin{equation}
    \phi_i = \frac{\epsilon \cdot u(\mathbf{x}_i, \mathbf{q})}{2\Delta u}.
    \label{eq:logits}
\end{equation}
We then select the indices corresponding to the largest $m$ values of the perturbed scores $s_i = \phi_i + G_i$, where $G_i \sim \text{Gumbel}(0, 1)$. This vectorized operation is numerically stable and scales linearly with the size of the candidate set $|\mathcal{S}|$.

\subsection{Theoretical Analysis}

We now provide proofs for the privacy and utility guarantees of the proposed mechanism. Throughout the analysis, we assume the feature space is bounded such that $\|\mathbf{x}\|_p \le 1$ for all $\mathbf{x}$. We analyze the privacy under the \textit{substitution adjacency} model, where two candidate sets $\mathcal{S}$ and $\mathcal{S}'$ differ by the value of exactly one sample (i.e., replacing $\mathbf{x}_k$ with $\mathbf{x}'_k$).

\subsubsection{Equivalence Lemma}

A prerequisite for our privacy analysis is the mathematical equivalence between the gumbel-top-$k$ procedure and the exponential mechanism.

\begin{lemma}[Exact Sampling Equivalence~\cite{vieira2014gumbel, kool2019stochastic}]
\label{lemma:equivalence}
Let $u_i$ be the utility score for the $i$-th candidate and $G_i \sim \text{Gumbel}(0, 1)$ be i.i.d. noise. Selecting the top-$m$ indices based on perturbed scores $u_i + G_i$ is equivalent to sampling $m$ indices sequentially without replacement from a categorical distribution proportional to $\exp(u_i)$. Thus, this procedure realizes the exponential mechanism over the domain of index subsets.
\end{lemma}

\subsubsection{Privacy Guarantee}

We analyze the privacy guarantee by formulating the neighbor selection as a mechanism $\mathcal{M}$ that outputs a subset of indices $\mathcal{I} \subset \{1, \dots, |\mathcal{S}|\}$ of size $m$.

\begin{theorem}[Privacy Guarantee]
The mechanism $\mathcal{M}$ that selects a subset of indices $\mathcal{I}$ with probability proportional to $\exp\left(\frac{\epsilon \cdot U(\mathcal{I}, \mathcal{S})}{2\Delta u}\right)$ satisfies $\epsilon$-differential privacy, where the set utility is defined as $U(\mathcal{I}, \mathcal{S}) = \sum_{i \in \mathcal{I}} u(\mathbf{x}_i, \mathbf{q})$.
\end{theorem}

\begin{proof}
The proof relies on the sensitivity analysis of the utility function with respect to the substitution of one candidate.

\noindent \textbf{Sensitivity Analysis:}
Let $\Omega$ be the set of all subsets of indices of size $m$. Consider two adjacent candidate sets $\mathcal{S}, \mathcal{S}'$ differing only at index $k$ (where $\mathbf{x}_k \in \mathcal{S}$ is replaced by $\mathbf{x}'_k \in \mathcal{S}'$). The sensitivity $\Delta U$ is defined as $\max_{\mathcal{S}, \mathcal{S}'} |U(\mathcal{I}, \mathcal{S}) - U(\mathcal{I}, \mathcal{S}')|$.
For any index subset $\mathcal{I} \in \Omega$:
\begin{itemize}
    \item Case 1: $k \notin \mathcal{I}$. The selected indices do not include the modified candidate. Thus, $U(\mathcal{I}, \mathcal{S}) - U(\mathcal{I}, \mathcal{S}') = 0$.
    \item Case 2: $k \in \mathcal{I}$. The subset includes the modified candidate. By the triangle inequality and the bounded feature space assumption, the change is:
    \[
    |u(\mathbf{x}_k, \mathbf{q}) - u(\mathbf{x}'_k, \mathbf{q})| \le \|\mathbf{x}_k - \mathbf{x}'_k\|_p \le 2.
    \]
\end{itemize}
Consequently, the sensitivity is bounded by $\Delta U = 2$.

\noindent \textbf{Privacy Loss Analysis:}
Based on Lemma \ref{lemma:equivalence}, our mechanism samples according to the exponential mechanism. For any output index subset $\mathcal{I} \in \Omega$, the probability ratio is bounded by the ratio of the unnormalized probabilities and the partition functions $Z_{\mathcal{S}} = \sum_{\mathcal{J}} \exp(\frac{\epsilon U(\mathcal{J}, \mathcal{S})}{2\Delta U})$:
\begin{equation}
{\small 
\begin{aligned}
    \frac{Pr[\mathcal{M}(\mathcal{S}) = \mathcal{I}]}{Pr[\mathcal{M}(\mathcal{S}') = \mathcal{I}]}
    &= \frac{\exp\left(\frac{\epsilon U(\mathcal{I}, \mathcal{S})}{2\Delta U}\right)}{\exp\left(\frac{\epsilon U(\mathcal{I}, \mathcal{S}')}{2\Delta U}\right)} \cdot \frac{Z_{\mathcal{S}'}}{Z_{\mathcal{S}}} \\
    &\le \exp\left(\frac{\epsilon |U(\mathcal{I}, \mathcal{S}) - U(\mathcal{I}, \mathcal{S}')|}{2\Delta U}\right) \cdot \exp\left(\frac{\epsilon \Delta U}{2\Delta U}\right) \\
    &\le \exp\left(\frac{\epsilon}{2}\right) \cdot \exp\left(\frac{\epsilon}{2}\right) = \exp(\epsilon).
\end{aligned}
}
\end{equation}

This concludes that $\mathcal{M}$ satisfies $\epsilon$-DP regarding the selection of neighbor indices.
\end{proof}

\subsubsection{Utility Guarantee}

Finally, we characterize the utility of the selected subset compared to the optimal subset.

\begin{theorem}[Utility Guarantee]
Let $\mathcal{I}_{OPT}$ be the subset with the maximum utility $OPT$. For any $t > 0$, the probability that the selected subset $\mathcal{I}$ has utility $U(\mathcal{I}) \le OPT - \frac{2\Delta u}{\epsilon}(\ln|\Omega| + t)$ is bounded by:
\begin{equation}
    Pr\left[ U(\mathcal{I}) \le OPT - \frac{2\Delta u}{\epsilon}(\ln|\Omega| + t) \right] \le e^{-t},
\end{equation}
where $|\Omega| = \binom{|\mathcal{S}|}{m}$ denotes the size of the output space (total number of possible subsets).
\end{theorem}

\begin{proof}
Let $\mathcal{B} = \{ \mathcal{I} \in \Omega \mid U(\mathcal{I}) \le c \}$ be the set of index subsets with utility below a threshold $c$. Following the standard utility analysis of the exponential mechanism \cite{dwork2014algorithmic}, the probability of selecting a subset from $\mathcal{B}$ is bounded by:
\begin{equation}\label{eq:bound}
    Pr[U(\mathcal{I}) \le c] \le \frac{|\mathcal{B}| \exp(\frac{\epsilon c}{2\Delta u})}{\exp(\frac{\epsilon OPT}{2\Delta u})} \le |\Omega| \exp\left( \frac{\epsilon (c - OPT)}{2\Delta u} \right).
\end{equation}
We wish to bound this probability by $e^{-t}$. Setting the RHS of (\ref{eq:bound}) to $e^{-t}$ and solving for $c$, we have
\begin{align*}
    |\Omega| \exp\left( \frac{\epsilon (c - OPT)}{2\Delta u} \right) &= e^{-t} \\
    \frac{\epsilon (c - OPT)}{2\Delta u} &= -t - \ln|\Omega| \\
    c &= OPT - \frac{2\Delta u}{\epsilon} (\ln|\Omega| + t).
\end{align*}
Substituting this value of $c$ back, we obtain the final bound:
\begin{equation}
    Pr\left[ U(\mathcal{I}) \le OPT - \frac{2\Delta u}{\epsilon} (\ln|\Omega| + t) \right] \le e^{-t}.
\end{equation}
This demonstrates that with high probability, the utility of the selected neighbor indices is close to that of the optimal $m$-nearest neighbors.
\end{proof}

\section{Evaluation}
\label{sec:evaluation}

In this section, we begin with a detailed description of the datasets used to rigorously test our approach, followed by a comprehensive description of the target model architectures and their characteristics. We then specify the configuration and procedures of our experimental setup. The core analysis presents membership inference attack results, comparing undefended models with those using our proposed defense to show its effectiveness in mitigating membership inference attacks (MIAs). Furthermore, we provide an in-depth utility analysis, demonstrating the impact of integrating our defense on the model's performance on its primary task. Finally, we benchmark our defense mechanism against established and state-of-the-art existing privacy-preserving techniques like MemGuard \cite{Jia2019memgaurd} and Differential Privacy Stochastic Gradient Descent (DP-SGD)~\cite{Abadi_2016} to contextualize the advantages and trade-offs of our approach within the broader landscape of privacy-preserving machine learning.

\subsection{Datasets}
\label{sec:eval_dataset}

\subsubsection{Nursery} 

The Nursery\footnote{\href{https://archive.ics.uci.edu/dataset/76/nursery}{https://archive.ics.uci.edu/dataset/76/nursery}} is a benchmark dataset in machine learning that contains 12,960 instances, each described by 8 categorical attributes. The target variable has five distinct classes, representing different outcomes for nursery school applications. All attributes are categorical, and the dataset is often used for classification tasks involving categorical data.

\subsubsection{Iris}

The Iris \footnote{\href{https://archive.ics.uci.edu/dataset/53/iris}{https://archive.ics.uci.edu/dataset/53/iris}} dataset contains 150 instances of iris flowers, described by four numeric attributes (all measured in centimeters). The target variable has three classes corresponding to the species. The dataset is balanced, with 50 instances per class, and is commonly used to evaluate classification algorithms.

\subsubsection{Adult (Census Income)}

The original Adult dataset\footnote{\href{https://archive.ics.uci.edu/dataset/2/adult}{https://archive.ics.uci.edu/dataset/2/adult}} comprises 48842 records with 14 attributes, including age, gender, and education. To develop an optimal machine learning model, the data is cleaned by addressing missing values, removing irrelevant columns, and eliminating duplicates. Numeric features are capped for outliers, log-transformed, and scaled, while categorical features are encoded. The binary classification task involves predicting whether a person earns over \$50K annually, encoded as 0/1.

\subsubsection{Purchase}

The purchase dataset employed in this research originates from Kaggle’s Acquire Valued Shoppers Challenge \footnote{\hyperlink{https://www.kaggle.com/c/acquire-valued-shoppers-challenge}{https://www.kaggle.com/c/acquire-valued-shoppers-challenge}} and contains the annual shopping records of thousands of individuals. We utilize a simplified, preprocessed version of this dataset provided by Shokri et al. \cite{Shokri2017}. This dataset comprises 197,324 samples, each described by 600 binary features. Each feature indicates whether a specific product was purchased by the individual. These samples are grouped into several different numbers of classes: 2, 10, 20, 50, and 100, which represent distinct purchase styles. The goal of the classification task is to predict an individual's purchase style based on their 600-dimensional binary feature vector.

\subsubsection{Location}

The dataset, derived from the Foursquare dataset \footnote{\hyperlink{https://sites.google.com/site/yangdingqi/home/foursquare-dataset}{https://sites.google.com/site/yangdingqi/home/foursquare-dataset}}, which contains location “check-in” records of several thousand individuals. The Location dataset, a simplified and preprocessed version by Shokri et al. \cite{Shokri2017}, contains 5,010 data samples, each described by 446 binary features. Similar to the purchase dataset, the task is a classification problem. The records are structured into 30 distinct clusters, making it a 30-class classification task. The objective is to predict the user's geo-social type based on their visit record.

\subsubsection{Texas}

The Texas dataset draws from the discharge data found in the Texas Hospital Inpatient Discharge Public Use Data File, provided by the Texas Department of State Health Services\footnote{\hyperlink{https://www.dshs.texas.gov/THCIC/Hospitals/Download.shtm}{https://www.dshs.texas.gov/THCIC/Hospitals/Download.shtm}}. It includes 925,128 patient records from 441 hospitals across Texas, covering the year 2006. Each record includes the patient’s demographic details. The machine learning task involves a 100-class classification problem predicting one of 100 surgical procedures from a patient’s health record.

\subsubsection{CIFAR}

This is a key benchmark dataset used to evaluate image recognition algorithms~\cite{krizhevsky2009learning}. It features 32×32 color images from 10 different classes, with 6000 images per class. We are loading it from a publicly available dataset collection called TensorFlow Datasets~\cite{TFDS}.

\subsection{Target Classifiers}

To evaluate the effectiveness and robustness of our proposed defense against membership inference attacks, we test it across various target models. We first consider classical machine-learning models and select three widely used classifiers: i) a Logistic Regression (LR) model configured with the lbfgs solver and a maximum of 10,000 iterations, ii) a Random Forest (RF) classifier with 100 decision trees, and iii) a Support Vector Classifier (SVC) using an  Radial Basis Function  (RBF) kernel. All classical models are taken from the Scikit-learn library \cite {scikit-learn}, a standard toolkit for machine-learning experimentation.

To complement these traditional models, we also evaluate our defense on a modern, high-capacity target: a feed-forward neural network trained using stochastic gradient descent (SGD). This network includes four fully connected hidden layers of sizes 1024, 512, 256, and 128, followed by a SoftMax output layer. We use the ReLU activation function throughout and train the model with a learning rate of 0.01.

For the CIFAR-10 image dataset, we evaluate our defense on a modern Convolutional Neural Network (CNN) architecture optimized for image recognition. The target model consists of two convolutional blocks: the first block contains two layers with 32 filters, and the second contains two layers with 64 filters. Each block is followed by a max-pooling layer and a dropout layer (rate of 0.25). The extracted features are processed through a fully connected dense layer of 256 units before reaching a SoftMax output layer. We utilize ReLU activation functions throughout the hidden layers and train the network using the Adam optimizer with a learning rate of 0.001.

\subsection{Evaluation Metrics}
\label{sec:eval_metrics}

The objective of the proposed defense mechanism is to mitigate the risk of membership inference attacks while maintaining the utility of the underlying machine learning model. Considering that the target models generate both a predicted label and an associated confidence vector, our defense mechanism is designed to preserve the predicted label: following smoothing, the class that originally predicted the highest probability continues to be the top prediction. Our approach guarantees zero label loss and ensures that the classifier's core functionality is preserved. 

\noindent We evaluate the proposed defense from two complementary perspectives, i.e., membership inference risk and model utility preservation. 

\subsubsection{Membership Inference Risk}

To assess the effectiveness of the proposed defense, we evaluate the model's vulnerability to various membership inference attacks. We measure attack accuracy both with and without the defense in place. This direct comparison allows us to quantify the proposed defense's success in reducing the membership information leaked via the model's confidence scores.

We explore two widely studied types of membership inference attacks: shadow-model attacks and metric-based attacks. These are the main approaches found in the research literature that offer a comprehensive evaluation of privacy leakage.

\textit{Shadow-model attacks}, initially introduced by Shokri et al.~\cite{Shokri2017}, involve training an attack classifier using machine learning to distinguish between members and non-members based on the output of the target model. The attacker constructs multiple shadow models using datasets from the same distribution as the target's training data. For each shadow model, the attacker gathers prediction vectors for both shadow-training samples (considered members) and shadow-test samples (considered non-members). These labeled predictions form an attack dataset used to train a binary classifier, known as an attack model. Once developed, the attack model takes a prediction vector from the target model and predicts the probability that the record was part of the target’s training set.

\textit{Metric-based attacks} infer membership by applying simple statistical tests to the target model’s prediction vector. Unlike shadow-model attacks, these attacks do not rely on training an attack classifier, making them lightweight and computationally inexpensive. We include the following representative metric-based attacks:

\begin{itemize}[noitemsep, topsep=0pt, partopsep=0pt, parsep=0pt]
    \item \textbf{Prediction confidence attack:} Yeom et al.~\cite{yeom2018} and Song et al.~\cite{Song2019} demonstrated that training samples generally attain higher confidence scores for their correct labels. An input is classified as a member if its predicted probability for the true label surpasses a predetermined threshold. Song et al.~\cite{Song2020SystematicEO} enhanced the attack methodology by incorporating class-dependent thresholds, where distinct threshold values are established for different class labels.
    
    \item \textbf{Prediction Entropy Attack:} Shokri et al.~\cite{Shokri2017} observe that training samples tend to have lower prediction entropy. Membership is inferred if the prediction entropy is below a threshold.
    
    \item \textbf{Modified Entropy Attack:} A refined entropy-based metric, proposed by Song et al~\cite{Song2020SystematicEO}, incorporates ground-truth awareness by enforcing monotonicity with respect to both correct and incorrect label probabilities.
\end{itemize}

Together, these attacks encompass a spectrum of potent and diverse adversarial strategies, spanning from high-capacity adversaries that rely on learned decision boundaries (shadow-model attacks) to lightweight adversaries that exploit statistical disparities in prediction vectors (metric-based attacks).

\subsubsection{Utility Preservation}

To evaluate the defense's impact on model outputs, we quantify the resulting distortion in the confidence vector by calculating the distance between the original and smoothed confidence vectors. This distortion is measured using two distinct metrics: 

\begin{itemize}[noitemsep, topsep=0pt, partopsep=0pt, parsep=0pt]
    \item \textbf{Predicted-class probability distortion}: This is the mean absolute difference between the original probability and the smoothed probability for the predicted class.
    \item \textbf{Confidence vector distortion:} This measures the mean Euclidean distance between the original predicted confidence vector and the smoothed confidence vector.
\end{itemize}
  
These two metrics collectively offer a comprehensive view of the defense's effectiveness, enabling us to assess both aspects of the privacy–utility trade-off.

\begin{table}[htbp]
\centering
\caption{Shadow-model based membership inference attack accuracy (with and without defense) across datasets and classical ML models.}
\label{tab:shadow_attack_results}
\small
\setlength{\tabcolsep}{4pt}
\renewcommand{\arraystretch}{0.9} 
\setlength{\aboverulesep}{1pt}
\setlength{\belowrulesep}{1pt}
\begin{tabularx}{\columnwidth}{l X c c}
\toprule
\textbf{Dataset} & \textbf{Model} & \textbf{Acc. (No Def.)} & \textbf{Acc. (With Def.)} \\
\midrule
Nursery & RF  & 0.6672 & 0.5372 \\
        & LR  & 0.4989 & 0.5020 \\
        & SVC & 0.4998 & 0.5014 \\
\midrule
Iris    & RF  & 0.7222 & 0.5556 \\
        & LR  & 0.6111 & 0.5556 \\
        & SVC & 0.4167 & 0.5000 \\
\midrule
Adult   & RF  & 0.5384 & 0.5468 \\
        & LR  & 0.4995 & 0.4994 \\
        & SVC & 0.5017 & 0.4993 \\
\midrule
Purchase-10  & RF  & 0.9679 & 0.4951 \\
             & LR  & 0.5156 & 0.5207 \\
             & SVC & 0.7714 & 0.5109 \\
\midrule
Purchase-20  & RF  & 0.9589 & 0.4719 \\
             & LR  & 0.4774 & 0.5078 \\
             & SVC & 0.9217 & 0.5101 \\
\midrule
Purchase-50  & RF  & 0.9677 & 0.5058 \\
             & LR  & 0.5165 & 0.5055 \\
             & SVC & 0.9125 & 0.4979 \\
\midrule
Purchase-100 & RF  & 0.9791 & 0.5126 \\
             & LR  & 0.5532 & 0.5597 \\
             & SVC & 0.8902 & 0.5117 \\
\midrule
Location-30     & RF  & 0.8952 & 0.4920 \\
             & LR  & 0.6461 & 0.5106 \\
             & SVC & 0.8643 & 0.5089 \\
\midrule 
Texas-100    & RF  & 0.8330 & 0.4940 \\
             & LR  & 0.5020 & 0.4940 \\
             & SVC & 0.5320 & 0.4950 \\
\bottomrule
\end{tabularx}
\end{table}

\begin{table}[htbp]
\centering
\caption{Shadow-model based membership inference attack accuracy (with and without proposed defense) against neural network models across datasets.}
\label{tab:shadow_attack_results_nn}
\small
\setlength{\tabcolsep}{4pt} 
\begin{tabularx}{\columnwidth}{@{} l X c c @{}}
\toprule
\textbf{Dataset} & & \textbf{Acc. (No Def.)} & \textbf{Acc. (With Def.)} \\
\midrule
Location-30  & & 0.766 & 0.497  \\
\addlinespace
Purchase-10 & & 0.577 & 0.504 \\
\addlinespace
Purchase-100   & & 0.571 & 0.495 \\
\bottomrule
\end{tabularx}
\end{table}

\begin{table}[htbp]
\centering
\caption{Shadow-model based membership inference attack accuracy (with and without proposed defense) against different models on CIFAR-10 dataset.}
\label{tab:shadow_attack_results_cifar}
\small
\setlength{\tabcolsep}{4pt} 
\begin{tabularx}{\columnwidth}{@{} l X c c @{}}
\toprule
\textbf{Model} & & \textbf{Acc. (No Def.)} & \textbf{Acc. (With Def.)} \\
\midrule
RF  & & 0.985 & 0.500  \\
\addlinespace
LR & & 0.539 & 0.500 \\
\addlinespace
SVC   & & 0.918 & 0.495 \\
\addlinespace
CNN   & & 0.707 & 0.618 \\
\bottomrule
\end{tabularx}
\end{table}


\subsubsection{Parameter Selection and Trade-offs} 

The neighborhood size, governed by the parameter $m$, functions as a balancing factor between membership privacy and model utility. A smaller neighborhood produces a smoothed output that closely resembles the original query, potentially preserving membership signals. Conversely, an excessively large neighborhood can substantially distort the confidence vector. Based on preliminary tuning across various datasets, we set $m=5$ as the default for densely populated datasets to ensure robust protection with minimal distortion. For sparse datasets, we recommend setting $m=3$ to maintain localized relevance. 

\subsection{Evaluation Results and Analysis}
\label{sec:eval}

We evaluate the performance of our proposed defense mechanism based on the aforementioned metrics: attack inference accuracy and utility loss. The empirical findings clearly demonstrate that the proposed confidence-smoothing defense significantly reduces the risk of membership inference across a variety of datasets, models, and attack modalities. It accomplishes this while inducing only controlled and often minimal distortions in predicted confidence vectors. The defense is particularly effective in high inference risk scenarios, such as with Random Forest and SVC models on high-dimensional datasets like \textit{Purchase-100} and \textit{Location-30}, where it consistently diminishes attack success rates to levels approaching or below the random guessing baseline. Conversely, it adopts a more conservative approach in low inference risk environments, where attack accuracy is already proximate to a random guess in the absence of any defense. 

\subsubsection{Defense Robustness Evaluation}

For the attack inference accuracies, we begin by calculating the attack accuracy of shadow-model based and metric-based attacks against various ML models across various datasets, both with and without our defense mechanism, see Tables~\ref{tab:shadow_attack_results},~\ref{tab:shadow_attack_results_nn},~\ref{tab:shadow_attack_results_cifar} \&~\ref{tab:metric_attack_results} for results. 

\noindent\textbf{Shadow-model based attacks:} The proposed defense produces the most significant absolute reductions in attack accuracy precisely where the undefended models are most susceptible to such attacks. For example, on the high-dimensional \textit{Purchase-10}, \textit{Purchase-20}, and \textit{Purchase-50} datasets, the undefended Random Forest models display exceptionally high average attack accuracies in the range of $0.96$--$0.97$. When employing our defense mechanism, these values decrease to approximately $0.47$--$0.51$, corresponding to absolute reductions of approximately $0.46$--$0.49$ in attack accuracy. Likewise, on the \textit{Location} dataset, the RF shadow attack accuracy declines from $0.8952$ to $0.4920$, representing a reduction of roughly $0.40$. A similar trend is observed for Support Vector Classifier (SVC) on these high-dimensional datasets. For instance, on it \textit{Purchase-20} and it \textit{Purchase-50}, the shadow model attack accuracies for SVC decline from $0.9217$ and $0.9125$ to approximately $0.51$ and $0.50$, respectively. Conversely, Logistic Regression models are often already close to random guess without any defense on datasets such as \textit{Nursery}, \textit{Adult}, and some \textit{Purchase} variants exhibit only slight fluctuations after the implementation of defenses, with variations of merely a few percentage points. This behavior is advantageous: when the baseline model leaks minimal membership information, the defense preserves the attack accuracy at a nearly constant level, thereby avoiding any unintended side effects. For low-dimensional datasets such as {Iris} and {Nursery}, the defense continues to offer significant improvements in specific configurations. For instance, the random forest shadow accuracy on {Iris} decreases from $0.7222$ to $0.5556$, and on {Nursery} from $0.6672$ to $0.5372$. These findings suggest that the proposed defense can nearly eliminate the strong membership signal leaked by the models, effectively reducing the shadow attacker's success rate to that of random guessing. Similar trend is observed in the attack accuracy drop on CIFAR-10 dataset (see Table~\ref{tab:shadow_attack_results_cifar}). 

After evaluating the performance of our defense mechanism on classical machine learning target model we move to deep neural networks trained using stochastic gradient descent (SGD). We simulate shadow model attacks on the SGD target model trained using \textit{Location-30,Purchase-10 \& Purchase-100} datasets. We compute attack accuracies with and without the proposed defense mechanism. Table~\ref{tab:shadow_attack_results_nn} shows that the defense mechanism effectively neutralizes membership risks for neural networks, successfully driving attack accuracies down from as high as $0.76$ to near-random guess levels ($\approx 0.50$). We also evaluate the performance of our defense mechanism on high-dimensional, non-tabular data i.e \textit{CIFAR-10} image dataset. We simulate shadow-model attacks against different models trained on this dataset. Table~\ref{tab:shadow_attack_results_cifar} illustrates the attack accuracies reduced from higher than $0.9$ to $\approx 0.50$ in the case of RF and SVC models.  These models include a high-dimensional Convolutional Neural Network (CNN) with an undefended attack accuracy of $0.707$, which is effectively reduced to $\approx 0.60$.

\begin{table}[htbp]
\centering
\caption{Metric-based membership inference attack accuracy (with and without defense) across datasets and classical ML models.  M: Model; PC: Prediction Confidence Attack; PCD: Prediction Confidence Attack with Defense; PE: Prediction Entropy Attack; PED: Prediction Entropy Attack with Defense; ME: Modified Entropy Attack; MD: Modified Entropy Attack with Defense.}
\label{tab:metric_attack_results}

\small

\begin{tabularx}{\columnwidth}{@{} l l *{6}{X} @{}}
\toprule
\textbf{Dataset} & \textbf{M} & \textbf{PC} & \textbf{PCD} & \textbf{PE} & \textbf{PED} & \textbf{ME} & \textbf{MD} \\
\midrule
Nursery & RF  & 0.55 & 0.50 & 0.76 & 0.49 & 0.76 & 0.52 \\
        & LR  & 0.48 & 0.50 & 0.70 & 0.59 & 0.70 & 0.64 \\
        & SVC & 0.49 & 0.49 & 0.48 & 0.41 & 0.48 & 0.41 \\
\midrule
Iris    & RF  & 0.58 & 0.56 & 0.72 & 0.56 & 0.72 & 0.56 \\
        & LR  & 0.50 & 0.56 & 0.69 & 0.64 & 0.69 & 0.67 \\
        & SVC & 0.55 & 0.57 & 0.61 & 0.50 & 0.58 & 0.56 \\
\midrule
Adult   & RF  & 0.53 & 0.51 & 0.57 & 0.50 & 0.59 & 0.52 \\
        & LR  & 0.50 & 0.51 & 0.51 & 0.50 & 0.51 & 0.51 \\
        & SVC & 0.50 & 0.51 & 0.51 & 0.50 & 0.51 & 0.51 \\
\midrule
Purchase-10    & RF  & 0.51 & 0.50 & 0.96 & 0.53 & 0.96 & 0.54 \\
        & LR  & 0.51 & 0.50 & 0.56 & 0.52 & 0.57 & 0.52 \\
        & SVC & 0.50 & 0.51 & 0.80 & 0.52 & 0.80 & 0.52 \\
\midrule
Purchase-20    & RF  & 0.51 & 0.53 & 0.97 & 0.51 & 0.98 & 0.54 \\
        & LR  & 0.50 & 0.51 & 0.56 & 0.51 & 0.60 & 0.52 \\
        & SVC & 0.49 & 0.52 & 0.81 & 0.52 & 0.81 & 0.53 \\
\midrule
Purchase-50    & RF  & 0.52 & 0.52 & 0.96 & 0.52 & 0.97 & 0.55 \\
        & LR  & 0.51 & 0.51 & 0.61 & 0.52 & 0.64 & 0.53 \\
        & SVC & 0.50 & 0.53 & 0.78 & 0.53 & 0.78 & 0.53 \\
\midrule
Purchase-100   & RF  & 0.62 & 0.50 & 0.89 & 0.52 & 0.90 & 0.53 \\
        & LR  & 0.50 & 0.50 & 0.52 & 0.50 & 0.54 & 0.52 \\
        & SVC & 0.51 & 0.51 & 0.79 & 0.54 & 0.80 & 0.51 \\
\midrule
Location-30 & RF  & 0.55 & 0.54 & 0.94 & 0.54 & 0.95 & 0.53 \\
        & LR  & 0.52 & 0.52 & 0.75 & 0.51 & 0.78 & 0.51 \\
        & SVC & 0.52 & 0.54 & 0.50 & 0.50 & 0.50 & 0.50 \\
\midrule
Texas-100   & RF  & 0.66 & 0.53 & 0.75 & 0.50 & 0.77 & 0.51 \\
        & LR  & 0.50 & 0.50 & 0.50 & 0.50 & 0.50 & 0.50 \\
        & SVC & 0.51 & 0.50 & 0.50 & 0.50 & 0.51 & 0.50 \\
\bottomrule
\end{tabularx}
\end{table}

\noindent\textbf{Metric-based attacks:} The results in Table~\ref{tab:metric_attack_results} illustrate that the defense mechanism does not target a specific attack type and provides robust mitigation against different types of metric-based attacks. It consistently diminishes statistical signatures that attackers depend upon. For Random Forest (RF) models on the \textit{Purchase-10}, \textit{Purchase-20}, \textit{Purchase-50}, and \textit{Location} datasets, the Prediction Entropy Attack (PE) and Modified Entropy Attack (ME) achieve high accuracy rates in the absence of defense, often ranging between $0.94$ and $0.98$. Upon the implementation of the defense, these accuracy levels decrease to approximately $0.51$ to $0.55$, nearing the level of random guessing. For example, on \textit{Purchase-20} with RF, the accuracies for PE and ME decline from $0.97$ and $0.98$ to $0.51$ and $0.54$, respectively. Similarly, on \textit{Location}, these accuracies decrease from $0.94$ and $0.95$ to $0.54$ and $0.53$. In cases where attack accuracies are close to the random baseline, even without any defense, and the defended variants remain similarly close (e.g., accuracies around $0.50$–$0.52$). 

Overall, these findings indicate that our defense mechanism provides extensive protection against both learned (shadow-model) and handcrafted (metric-based) attack techniques, ensuring robustness irrespective of whether the attacker employs a different attack model or merely relies on confidence metrics.

\begin{table}[htbp]
\centering
\caption{Predicted-class confidence difference (PCD) and confidence vector difference (CVD) across datasets and classical ML models.}
\label{tab:confidence_differences}
\small

\begin{tabularx}{\columnwidth}{@{} X X X X @{}}
\toprule
\textbf{Dataset} & \textbf{Model} & \textbf{PCD} & \textbf{CVD} \\
\midrule
Nursery & RF  & 0.075 & 0.105 \\
        & LR  & 0.047 & 0.067 \\
        & SVC & 0.029 & 0.041 \\
\midrule
Iris    & RF  & 0.061 & 0.083 \\
        & LR  & 0.040 & 0.056 \\
        & SVC & 0.047 & 0.067 \\
\midrule
Adult   & RF  & 0.080 & 0.113 \\
        & LR  & 0.051 & 0.073 \\
        & SVC & 0.039 & 0.055 \\


\midrule
Purchase-10  & RF  & 0.340 & 0.369 \\
             & LR  & 0.016 & 0.023 \\
             & SVC & 0.095 & 0.119 \\
\midrule
Purchase-20  & RF  & 0.326 & 0.344 \\
             & LR  & 0.014 & 0.020 \\
             & SVC & 0.121 & 0.145 \\
\midrule
Purchase-50  & RF  & 0.299 & 0.310 \\
             & LR  & 0.026 & 0.036 \\
             & SVC & 0.164 & 0.186 \\
\midrule
Purchase-100 & RF  & 0.249 & 0.256 \\
             & LR  & 0.030 & 0.041 \\
             & SVC & 0.098 & 0.122 \\
\midrule
Location-30     & RF  & 0.313 & 0.329 \\
             & LR  & 0.045 & 0.060 \\
             & SVC & 0.278 & 0.299 \\
\midrule
Texas-100    & RF  & 0.347 & 0.391 \\
             & LR  & 0.007 & 0.013 \\
             & SVC & 0.006 & 0.016 \\
\bottomrule
\end{tabularx}
\end{table}

\subsubsection{Utility Analysis}

We compute both predicted-class probability distortion (the absolute difference between the baseline and smoothed (defended) predicted class probabilities) and confidence vector distortion (which is the euclidean distance between baseline and smoothed (defended) confidence vector for all the classes). Table~\ref{tab:confidence_differences} illustrates the comparisons of predicted-class probability distortion (PCD) and confidence vector distortion (CVD) across different datasets and target models, demonstrating the extent of the defense mechanism's impact on the outputs of the models. For all classical models and datasets examined, the predicted-class probability distortion (PCD) ranges from $0.014$ to $0.340$, while the confidence vector distortion (CVD) spans from $0.020$ to $0.369$. On low-dimensional datasets such as \textit{Nursery}, \textit{Iris}, and \textit{Adult}, the observed distortions are comparatively minor: Random Forest (RF) models exhibit PCD values approximately between $0.06$ and $0.08$, and CVD values around $0.08$ to $0.11$. Conversely, Logistic Regression (LR) and Support Vector Classifier (SVC) models demonstrate even smaller distortions, with PCD primarily below $0.05$ and CVD below $0.08$. These findings indicate that the defense mechanism can substantially diminish membership signals, as demonstrated by reduced shadow-model and metric-based attack accuracies, while altering the primary class prediction probabilities by less than 10\% on average. For the high-dimensional datasets, such as \textit {Purchase} and \textit{Location}, particularly with Random Forest (RF) and Support Vector Classifier (SVC), the defense mechanism results in greater distortions. For instance, RF on \textit{Purchase-10} exhibits a PCD of $0.340$ and a CVD of $0.369$, whereas on \textit{Location}, the PCD is $0.313 $and the CVD is $0.329$. This observation aligns with the presence of stronger baseline leakage: reducing attack accuracy from over 0.9 to approximately 0.5 on these models necessitates more rigorous smoothing of confidence vectors.

\begin{figure*}[t]
    \centering
    \includegraphics[width=0.8\textwidth]{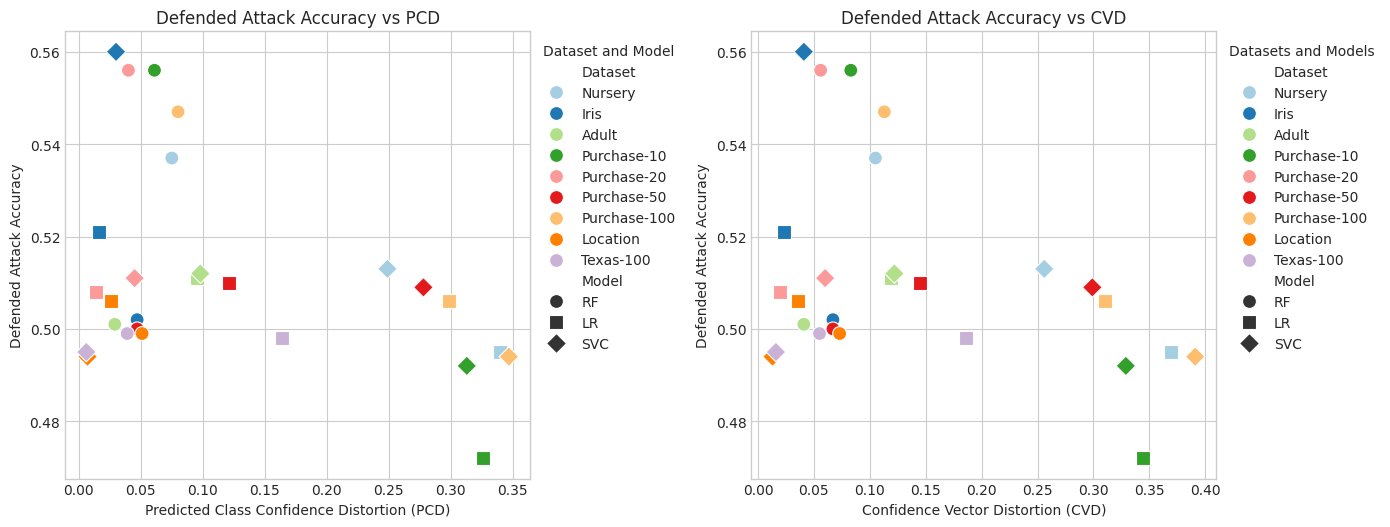}
    \caption{Correlation between attack accuracy with Neighbor defense and the proposed distortion metrics. The left plot shows the attack accuracy versus PCD, and the right plot shows the attack accuracy versus CVD across various ML models and datasets.}
    \label{fig:PCD_CVD_attack_acc}
\end{figure*}

\subsubsection{Correlation Analysis}

We first examined the relationship between defense robustness and its associated cost (PCD and CVD). Figure~\ref{fig:PCD_CVD_attack_acc} show scatter plots of attack accuracy (with neighborhood defense) versus PCD and CVD across various models and datasets. Scatter plots do not reveal any direct corelation between the defended attack accuracies and the distortion. Instances with defended attack accuracies closer to 0.5 have PCD and CVD values across the spectrum. This indicates towards the adaptive nature of the defense mechanism, modulating the level of smoothing based on the specific membership leakage. To further explore this, we examine the relationship between defense aggressiveness and the associated cost, i.e., CVD. The defense aggressiveness is measured by the attack accuracy drop, which is the difference between attack accuracy with and without defense mechanism. We compute this across various models and datasets. 
Figure~\ref{fig:shadow_attack_drop_heatmap} presents a heatmap depicting the shadow-model attack accuracy drop across different models and datasets when employing the proposed defense mechanism. 
Conversely, Figure~\ref{fig:shadow_CVD_heatmap} displays a heatmap illustrating the extent of CVD under the same conditions. Both plots show hot spots at similar positions, suggesting a direct correlation between the drop in attack accuracy and CVD. Specifically, a higher accuracy drop is associated with increased confidence vector distortion, and vice versa. This observation underscores the adaptability of the defense mechanism across diverse scenarios. In most cases, the defended attack accuracy approaches that of random guessing; however, when the undefended model is at high risk, the defense mechanism applies aggressive smoothing to achieve the desired effect, resulting in high CVD. Conversely, in low-inference-risk scenarios, the mechanism applies minimal smoothing, resulting in lower CVD. 

The key observation is that the defense mechanism maintains zero label loss: the predicted label remains consistent, thereby ensuring that primary task accuracy is entirely unaffected by the smoothing process. Consequently, the utility cost is confined to the distortion of the confidence vector, where PCD and CVD serve as a tolerance metric for deployment across diverse domains. The mechanism follows a "pay-as-you-go" approach, sacrificing utility only in direct proportion to the inherent leakage of the target model. This ensures that well-generalized models are not penalized with unnecessary noise, while the "membership signals" of more vulnerable models are effectively obfuscated. Such adaptability renders it a practical post-hoc solution suitable for various MLaaS environments with varying levels of model leakage.

\begin{figure*}[t] 
    \centering
    \begin{subfigure}{0.48\textwidth}
        \centering
        \includegraphics[width=\linewidth]{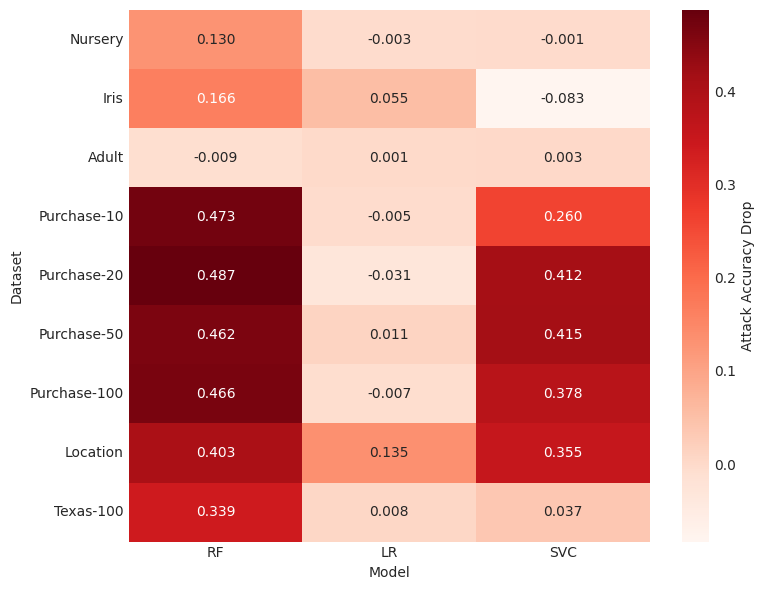}
        \caption{Absolute reduction in shadow-model attack accuracy.}
        \label{fig:shadow_attack_drop_heatmap}
    \end{subfigure}
    \hfill 
    \begin{subfigure}{0.48\textwidth}
        \centering
        \includegraphics[width=\linewidth]{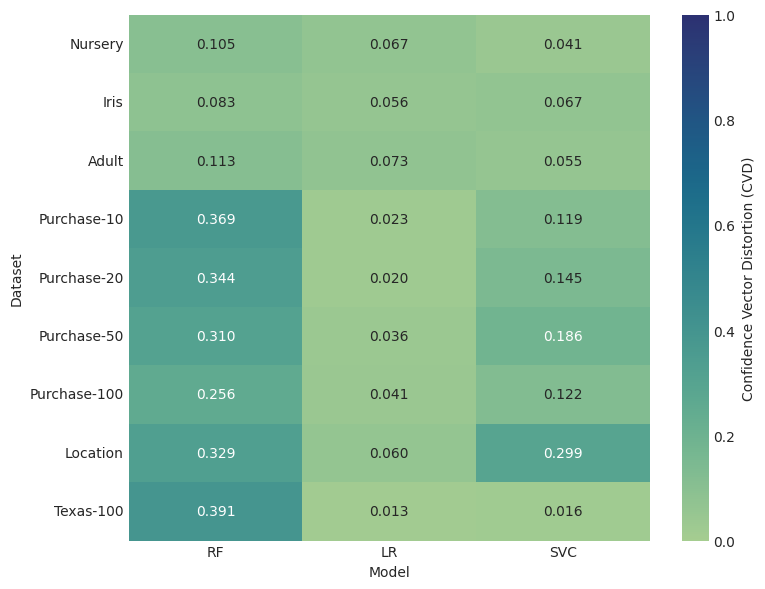}
        \caption{Confidence Vector Distortion (CVD) across test samples.}
        \label{fig:shadow_CVD_heatmap}
    \end{subfigure}
    
    \caption{Comparative heatmap analysis demonstrating the correlation between defensive effectiveness (accuracy drop) and the resulting utility distortion (CVD) across all evaluated datasets and classical ML classifiers.}
    \label{fig:combined_shadow_heatmaps}
\end{figure*}

\begin{table*}[t]
\centering
\small
\caption{Qualitative comparison of defenses mechanisms.}
\label{tab:qualitative_analysis}
\begin{tabular}{l c c c c c}
\toprule
\textbf{Defense} & \textbf{Retraining} & \textbf{Access Level} & \textbf{Latency Overhead} & \textbf{Zero Label Loss} & \textbf{Utility Loss}\\
\midrule

DP-SGD & Yes & White-box & High & Yes & High \\
\midrule

MemGuard & No & Black-box & Moderate & No & Moderate \\
\midrule

Neighborhood-Defense & No & Black-box & Low & No & Low \\

\bottomrule
\end{tabular}
\end{table*}

\begin{table*}[t]
\centering
\small
\caption{Predicted-class confidence distortion and confidence vector distortion for \textit{Purchase-10} dataset of SGD model and DP-SGD models.}
\label{tab:sgd_dpsgd_confidence_differences}
\begin{tabular}{l l c c}
\toprule
\textbf{Model Comparison} & \textbf{Metric} & \textbf{SGD vs DP-SGD} & \textbf{SGD vs SGD\_Defended} \\
&  &   & \textbf{(Neighborhood Blending)} \\
\midrule

DP-SGD ($\epsilon$ 65.15) 
& Predicted Class - Confidence Difference (PCD) & 0.0171 & 0.0056 \\
& Confidence Vector Difference (CVD) & 0.0241 & 0.0078 \\
\midrule

DP-SGD ($\epsilon$ 204.82)
& Predicted Class - Confidence Difference (PCD) & 0.0589 & 0.0037 \\
& Confidence Vector Difference (CVD) & 0.0807 & 0.0052 \\
\bottomrule
\end{tabular}
\end{table*}

\subsection{Comparative Analysis}

The evaluation results in Section \ref{sec:eval} demonstrate that our neighborhood confidence smoothing defense provides a robust and effective safeguard against membership inference attacks. In this section, we discuss the comparative performance of our mechanism relative to other existing approaches, such as MemGuard and DP-SGD.

\subsubsection{Qualitative Analysis}

Table~\ref{tab:qualitative_analysis} provides a qualitative comparison of the proposed defense mechanism against two established baselines: the training-time DP-SGD~\cite{Abadi_2016} and the post-hoc MemGuard~\cite{Jia2019memgaurd}. 
DP-SGD is incompatible with pre-trained models, requiring full retraining and white-box access to the model's internal parameters. It incurs high computational cost due to per-sample gradients and noise injection~\cite{Abadi_2016}, and suffers from high utility loss through both confidence vector distortion and label loss (see Table~\ref{tab:sgd_dpsgd_confidence_differences}). In contrast, MemGuard is a post-hoc solution that requires no retraining or internal model access and guarantees zero label loss. However, it requires training of an adversarial noise generator to optimize noise injection for robustness against inference attacks. Consequently, this incurs noticeable latency and results in moderate to high distortion of the confidence vectors, depending on the level of privacy desired~\cite{Jia2019memgaurd}.

In contrast, our proposed Neighborhood Blending mechanism is a truly lightweight, flexible solution that integrates seamlessly with existing systems. It can be easily applied to any pre-trained classifier without requiring full model retraining or white-box access. Plus, unlike other post-hoc methods like MemGuard, it doesn't require auxiliary training or repetitive noise optimization. Acting as a model-agnostic wrapper, it offers broad compatibility and ensures zero label loss while keeping confidence vector distortion minimal. This combination of high privacy effectiveness and minimal computational overhead makes Neighborhood Blending an excellent choice for real-world applications on high-utility systems.

\subsubsection{Quantitative Analysis}
For quantitative analysis, we compare the performance of DP-SGD with our mechanism on the high-dimensional dataset, i.e., \textit{Purchase-10}. The undefended target model, i.e., an SGD neural network, achieves a training accuracy of $1.00$ and a test accuracy of $0.9265$. Similarly, we train two DP-SGD versions of the same model at different epsilon ($\epsilon$) values. While training-time defenses like DP-SGD reduce this risk to near-random levels, they incur severe utility penalties. The first DP-SGD model achieves a training accuracy of $0.5844$ and a test accuracy of $0.4975$ at an $\epsilon$ value of $65.15$. The second achieves training and test accuracies of $0.7686$ and $0.7990$, respectively, at an $\epsilon$ value of $204.82$. This significant reduction in model performance highlights a substantial predictive label loss, which our proposed neighborhood-based mechanism eliminates by guaranteeing zero label loss. Furthermore, as shown in Table~\ref{tab:sgd_dpsgd_confidence_differences}, our method results in PCD and CVD values that are orders of magnitude lower than those introduced by even the high $\epsilon$ DP-SGD model, while achieving the same level of protection against membership inference risk as DP-SGD. This confirms that our post-hoc defense provides equivalent privacy protection to state-of-the-art training-time mechanisms while offering significantly superior utility preservation.

MemGuard provides protection by introducing adversarial noise tailored to a specific privacy budget ($\epsilon$), which determines the overall level of defense effectiveness. This noise is specifically designed to be adversarial against shadow-model attacks, effectively misdirecting the attack classifier's decision boundaries. However, while MemGuard is effective against shadow-model adversaries, it remains vulnerable to metric-based attacks such as confidence or entropy thresholds. A systematic evaluation by Song et al~\cite{Song2020SystematicEO} demonstrated this limitation, showing that MemGuard falls short against metric-based benchmark attacks, which achieved accuracies of  $0.691$ and $0.742$ on \textit{Location-30} and 
\textit{Texas-100} datasets, respectively. Although MemGuard obfuscates the membership signals within confidence vectors, it fails to alter the underlying statistical properties exploited by such attacks. Furthermore, Jia et al \cite{Jia2019memgaurd} reported that MemGuard achieves an average confidence vector distortion above $0.3$ for \textit{Location-30} dataset and $0.2$ for \textit{Texas-100} dataset, when achieving the inference accuracy of $0.5$. 

In contrast, our proposed neighborhood-based approach addresses these identified gaps by providing uniform robustness and superior utility. Specifically, it offers consistent protection across both learned and metric-based attack modalities. Furthermore, the mechanism achieves the same level of privacy protection with significantly lower distortion, maintaining the statistical integrity of the target model's outputs.

\section{Conclusion}
\label{sec:conclusion}

In this paper, we propose Neighborhood Blending, a lightweight, post-hoc mechanism to defend against blackbox membership inference attacks (MIAs). Neighborhood Blending effectively neutralizes MIAs by smoothing confidence vector outputs using a differentially private neighborhood sampling approach that renders member and non-member samples statistically indistinguishable from each other. Our empirical evaluation results show that Neighborhood Blending can effectively defend against MIAs and reduce attack accuracies to near-random levels across various datasets and model architectures, all while guaranteeing zero label loss. Following an adaptive "pay-as-you-go" distortion model, the defense modulates noise based on the model’s inherent membership inference risk, incurring orders of magnitude less distortion than training-time alternative defense mechanisms. Ultimately, Neighborhood Blending provides a practical, plug-and-play solution for securing MLaaS deployments without significant utility loss.

\cleardoublepage
\bibliographystyle{plain}
\bibliography{bib}

\end{document}